\documentclass[letterpaper, 10 pt, conference]{ieeeconf}  
\IEEEoverridecommandlockouts                              

  \let\labelindent\relax
\usepackage{amsmath} 
\usepackage{amssymb}  
\usepackage{tabulary}
\usepackage[english]{babel}
\usepackage{blindtext}
\usepackage[font=footnotesize]{caption}
\usepackage[makeroom]{cancel}
\usepackage{amsfonts}
\usepackage{subfig,pgfplots}
\usepackage{amsthm}  
\usepackage{amssymb}
\usepackage{graphicx}
\usepackage{pgf,tikz}
\usepackage{multirow}
\usepackage{multicol}
\usepackage{siunitx}
\renewcommand{\eqref}[1]{Eq.~(\ref{#1})}  
\usepackage{tabularx}
\usepackage{float}
\usepackage{cite}
\usepackage{enumitem} 
\usepackage{bbm} 

\usepackage[strict]{changepage}

\newcommand{\vect}[1]{\boldsymbol{#1}}
\newcommand{\mat}[1]{\boldsymbol{#1}}

\DeclareMathOperator{\diag}{\text{diag}}

\newtheorem{remark}{Remark}
\newtheorem{theorem}{Theorem}

\newtheorem{proposition}{Proposition}

\newtheorem{problem}{Problem}

\allowdisplaybreaks

\title{\LARGE \bf
A Quantum-Compliant Formulation for Network Epidemic Control
}

\author{Lorenzo Zino, Mattia Boggio, Deborah Volpe,  Giacomo Orlandi, Giovanna Turvani, and Carlo Novara 
\thanks{L.Zino, M. Boggio, G. Orlandi, G. Turvani, and C. Novara are with the Department of Electronics and Telecommunications, Politecnico di Torino, Italy (\texttt{\{lorenzo.zino,mattia.boggio,giacomo.orlandi, giovanna.turvani,carlo.novara\}@polito.it}). D. Volpe is with the National Institute of Geophysics and Vulcanology, Rome, Italy (\texttt{deborah.volpe@ingv.it}). This research was supported by the Department of Electronics and Telecommunications, Politecnico di Torino,  project ``ALLOY: Quantum Control of Complex Network Systems." The authors thank CINECA for the collaboration and access to their machines. 
 }%
}

\begin{document}

\maketitle
\thispagestyle{empty}

\begin{abstract}
We deal with controlling the spread of an epidemic disease on a network by isolating one or multiple  locations by banning people from leaving them. To this aim, we build on the susceptible--infected--susceptible and the susceptible--infected--removed discrete-time network  models, encapsulating a control action that captures mobility bans via removing links from the network. Then, we formulate the problem of optimally devising a control policy based on mobility bans that trades-off the burden on the healthcare system and the social and economic costs associated with interventions. The binary nature of mobility bans hampers the possibility to solve the control problem with standard optimization methods, yielding a NP-hard problem. Here, this is tackled by deriving a Quadratic Unconstrained Binary Optimization (QUBO) formulation of the control problem, and leveraging the growing potentialities of quantum computing to efficiently solve it. 
\end{abstract}

\section{Introduction}\label{sec:intro}

The recent COVID-19 health crisis has demonstrated how mathematical modeling and control-theoretic techniques can be powerful assets in the design, evaluation, and optimization of interventions policies during an epidemic outbreak~\cite{DellaRossa2020,Carli2020,Giordano2021}. Of particular interest are network epidemic models, thanks to their ability to capture the complex and heterogeneous patterns of contagion across a geographic region~\cite{Nowzari2016,Mei2017,Pare2020review,zino2021survey,Ye2023competitive}. 

In this context, different questions have been investigated, including how to distribute drugs~\cite{Drakopoulos2014,Nowzari2015,Mai2018,Somers2021,Walsh2025}, plan a vaccination campaign~\cite{Preciado2013,parino2021}, and how guide a collective behavioral response~\cite{alutto2021,Martins2023,Maitra2024,Parino2024,Bizyaeva2024}. Among this broad array of questions, a problem of paramount importance is to design intervention policies to control an epidemic outbreak, when no pharmaceutical interventions are available. In this scenario, control actions must rely on non-pharmaceutical interventions, including implementing lockdowns and travel bans. 

Controlling network epidemic processes via non-pharmaceutical interventions is typically a complex problem from the computational point of view. In fact, non-pharmaceutical intervention policies are typically of discrete nature. For instance, deciding whether banning mobility between two locations in a geographic network or not is a binary decision variable. Similar, lockdown can be implemented at different (but finite) levels. Hence, the mathematical formulation of optimal control problems for such network epidemic models typically yields binary (or integer) optimization problems that are difficult to be solved in an exact manner with classical optimization methods~\cite{vanmieghem2011}. Indeed, several heuristics have been developed, especially when one is interested in guaranteeing long-term eradication of the disease~\cite{Enns2012}, but only a minority of the approaches focus on the transient behavior of the epidemic outbreak, limiting their applicability in real-life applications. 

This calls for the exploration of novel optimization techniques to solve optimal control problems for the transient behavior of network epidemic processes. Traditional optimization approaches, while powerful, often face scalability limitations and computational bottlenecks when applied to such problems, especially under time constraints or in large-scale networks.
 In these regards, a promising research field that is recently emerging for the solution of complex and previously intractable optimization problems is quantum computing~\cite{King2021,Denchev2016,Berberich2024}. Quantum algorithms, such as the Quantum Approximate Optimization Algorithm (QAOA), Grover Adaptive Search (GAS), or quantum annealing (QA)~\cite{kadowaki1998quantum, ManufacturedSpins,Volpe2025}, offer new computational paradigms that exploit quantum mechanical principles like superposition and entanglement to explore large solution spaces more efficiently than classical methods. In this work, we consider quantum annealers as the target quantum hardware, as current quantum circuit-based devices are too noisy to reliably execute optimization algorithms. Moreover, simulating such devices on classical hardware imposes severe limitations on the maximum problem size that can be explored. 


A common feature of these approaches is their ability to efficiently solve problems that are formulated as Quadratic Unconstrained Binary Optimization (QUBO) problems~\cite{Volpe2025}. The QUBO formulation acts as a unifying framework, allowing a wide class of optimization problems ---including those arising in nonlinear model predictive control~\cite{novara2024quantum}, scheduling~\cite{marchioli2024scheduling}, and resource allocation~\cite{volpe2024quantum}--- to be mapped onto hardware-compliant models that are solvable by quantum or quantum-inspired devices.

Motivated by these emerging techniques, we take a step towards their integration in the problem of controlling network epidemics. We consider the two most fundamental models of epidemic progression on networks~\cite{zino2021survey}: susceptible--infected--susceptible (SIS) and susceptible--infected--removed (SIR) models, which capture diseases that do not provide immunity after recovery ---e.g., most sexually transmitted diseases (STIs)~\cite{Yorke1978}--- and diseases that provide permanent (or long-lasting) immunity, respectively. Then, we encapsulate a control action into the model, consisting in isolating one or multiple nodes and we formulate an optimal control problem for the transient behavior of network epidemic processes in order to design the control action in an effective way, trading-off healthcare costs associated to the spread of the disease and the social and economic impact of node isolation. 

Besides formulating the control problem, our  contribution is twofold. First, we derive a QUBO formulation for our optimal control problem. Second, leveraging this formulation, we explore the potentiality of quantum computing in solving the control problem. In particular, we consider a case study based on the spread of an infectious disease in Italy, with realistic model parameters and network structure. We demonstrate how the control problem, formulated as a QUBO problem, can be efficiently addressed using a quantum annealer~\cite{kadowaki1998quantum, ManufacturedSpins}, outperforming classical methods in terms of  computational time, with comparable solutions quality. 


\section{Controlled Network Epidemic Models}\label{sec:model}


We denote by $\mathbb R$, $\mathbb R_{\geq 0}$,  $\mathbb R_{> 0}$,  $\mathbb Z_{\geq 0}$, and $\mathbb Z_{> 0}$ the real, real nonnegative, strictly positive real, positive integer, and strictly positive integer numbers, respectively. Given $n,m\in\mathbb Z_{>0}$, bold lowercase font denote a vector $\vect{x}\in\mathbb R^n$, with $x_{i}$ its $i$th entry; and bold capital font denote a matrix $\mat{A}\in\mathbb R^{n\times m}$, with $A_{ij}$ the $j$th entry of its $i$th row. The all-$1$ vector and the identity matrix are denoted by $\vect 1$ and $\mat I$, respectively.

\subsection{Network population model}

We consider $n$ individuals, partitioned into $M$ locations, denoted by $\mathcal L:=\{1,\dots,M\}$. In location $i$, the total population is equal to $n_i$ individuals. In this paper, we will focus on large populations. Hence, for the sake of simplicity, it is reasonable to approximate the population as a continuum, i.e., $n_i\in\mathbb R_{>0}$. 
Individuals from different locations can come in contact due to their mobility. To capture this feature, we introduce a weighted graph $\mathcal G=(\mathcal L,\mat{A})$, where 
the weight matrix $\mat{A}\in\mathbb R_{\geq 0}^{M\times M}$ measures the level of interactions between individuals in different locations: $A_{ij}$ represents the amount of such interaction (relative to the interactions within each location). While not strictly necessary, it is  reasonable to assume that $A_{ij}\in[0,1]$. Furthermore, we assume that  $\mathcal G$ has no self-loops, yielding $A_{ii}=0$ for all $i\in\mathcal L$.


\subsection{Network SIS and SIR epidemic models}

We consider two fundamental models of epidemic progression: the SIS and the SIR models~\cite{Mei2017,zino2021survey}. In both models, individuals can be characterized by two possible health conditions: \emph{susceptible} (S) to the disease, or \emph{infected} (I) with the disease. The two models differ in whether individuals acquire immunity after recovery. In the SIS model, no immunity is acquired, and recovered individuals become immediately susceptible to the disease again (e.g., many  STIs~\cite{Yorke1978}). In the SIR model, instead, individuals become permanently immune after recovery. This is a good proxy also for scenarios where immunity wanes at a slower time-scale than the one of an epidemic wave (e.g., COVID-19). For the SIR model, a third health state, termed \emph{removed} (R), is used to represent this health condition. 

For each location $i\in\mathcal L$, let us denote by $x_i(t)\in[0,n_i]$ the number of infected individuals at time $t\in\mathbb Z_{\geq 0}$ in location $i\in\mathcal L$. In the network SIS model, the number of susceptible individuals is equal to $n_i-x_i(t)$. Hence, the health state of the whole system is fully determined by the \emph{state vector} $\vect{x}(t)=[x_1(t),\dots,x_M(t)]^\top$. For a generic location $i$, the number of infected individuals $x_i(t)$ is updated at each discrete time step according to two contrasting mechanisms: i) \emph{recovery}, which  is regulated by the \emph{recovery rate} $\mu\in\mathbb R_{\geq 0}$, whereby, at each discrete time step, a fraction $\mu$ of the infected individuals recover; ii) \emph{contagion}, whereby susceptible individuals who interact with infected individuals (in the same location or in other locations) become infected with \emph{infection rate} $\lambda\in\mathbb R_{\geq 0}$. These two contrasting mechanisms yield the following recursive equation:
\begin{equation}\label{eq:sis}
    x_i(t+1)=(1-\mu)x_i(t)+\frac{\lambda}{n_i}\left(n_i-x_i(t)\right)\alpha_i(t),
\end{equation}
for all $i\in\mathcal L$, where 
\begin{equation}\label{eq:alpha}
    \alpha_i(t)=x_i(t)+\sum\nolimits_{j\in\mathcal L} A_{ij}x_j(t)
\end{equation}
quantifies the \emph{infection force} in location $i$ at time $t$. 
\eqref{eq:alpha} comprises two terms: the first captures the contribution of infected individuals in the same location, the second term captures the impact of mobility to different locations. 


The value of the parameters depends on the time-step. The recovery rate $\mu$ can be interpreted as the fraction of infected individuals who recover in a time step or (equivalently) the inverse of the mean duration of the disease (in time-steps). 

\begin{remark}\label{rem:para}
The infection rate $\lambda$ is related to the well-known concept of the basic reproduction number, whereby one can express $\lambda\propto R_0\mu$, where the proportionality coefficient is the largest eigenvalue of matrix $\mat A+\mat I$~\cite{zino2021survey}. Hence, both parameters $\mu$ and $\lambda$ scale with the duration of a time-step.\end{remark}

A consequence is that, by setting a sufficiently small time-step, it is always possible to guarantee that  \eqref{eq:sis} is well-defined, as proved in the following statement.

\begin{proposition}\label{prop:invariance}
    If $\lambda\leq \max_{i\in\mathcal L}\big(1+\sum_{j\in\mathcal L} A_{ij}\frac{n_j}{n_i}\big)^{-1}$, then $\mathcal D=\prod_{i\in\mathcal L} [0,n_i]$ is positively invariant under \eqref{eq:sis}. 
\end{proposition}
\begin{proof}
Consider a generic $i\in\mathcal L$ and $t\geq 0$. From \eqref{eq:sis},  if $x_i(t)\in[0,n_i]$, then  $x_i(t+1)\geq(1-\mu)x_i(t)\geq 0$. On the other hand, since $\lambda\leq \big(1+\sum_{j\in\mathcal L} A_{ij}\frac{n_j}{n_i}\big)^{-1}$, then $\alpha_i(t)\leq n_i+\sum_{j\in\mathcal L} A_{ij}n_j\leq \frac{n_i}{\lambda}$,  implying $x_i(t+1)\leq(1-\mu)x_i(t)+(n_i-x_i(t))\leq n_i$, yielding the claim. 
\end{proof}

In the SIR model, for each location, a second state variable should be defined to account for the number of removed individuals, $y_i(t)\in[0,n_i]$. In this case, individuals become removed after recovery and they cannot contract the disease again. Hence,  the number of susceptible individuals is equal to $n_i-x_i(t)-y_i(t)$. In this scenario, the state of the system is characterized by the two vectors $\vect{x}(t)$ and $\vect{y}(t)=[y_1(t),\dots,y_M(t)]^\top$, whose components are updated according to the following dynamics:
\begin{equation}\label{eq:sir}
\begin{array}{l}
   \displaystyle x_i(t+1)=(1-\mu)x_i(t)+\frac{\lambda}{n_i}\left(n_i-x_i(t)-y_i(t)\right)\alpha_{i}(t),\\
    y_i(t+1)=y_i(t)+\mu x_i(t),
    \end{array}
\end{equation}
for all $i\in\mathcal L$, where the first equation differs from \eqref{eq:sis} only by the definition of susceptible population, which is now equal to $n_i-x_i(t)-y_i(t)$, and by the presence of an additional equation to account for recovered individuals who become removed. The observations in Remark~\ref{rem:para} and Proposition~\ref{prop:invariance} hold true also for \eqref{eq:sir}, implying that a thoughtful choice of the time-step guarantees that the discrete-time network SIR model is well-defined.


\subsection{Control}\label{sec:control}

In this paper, we focus on epidemic diseases for which no vaccines or effective pharmaceutical interventions are available. This is the case, e.g., of many STIs or for newly emerged diseases, as it was for the first wave of the COVID-19 pandemic. In such a scenario, the only intervention policies that can be enacted are non-pharmaceutical. Besides incentivizing the use of personal protection equipments (e.g., condoms for STIs or face masks for air-borne diseases), policy makers can enact intervention policies based on isolating one or multiple locations, not allowing the population of a location to leave that location. These interventions were implemented, e.g., during the first waves of the COVID-19 pandemic in many countries.

Technically, these interventions can be encapsulated into the network epidemic model by defining a set of binary control variables $\vect u=[u_1,\dots,u_M]^\top\in\{0,1\}^M$, each entry associated with a location, such that
\begin{equation}\label{eq:control}
    u_i=\left\{\begin{array}{ll}1&\text{if location $i$ is isolated,}\\
    0&\text{otherwise.}
    \end{array}\right.
\end{equation}
Then, for a controlled network epidemic model, the infection force in \eqref{eq:alpha} is replaced by the following expression, which accounts for isolation of one or multiple locations:
\begin{equation}\label{eq:alpha_control}
    \alpha_i(\vect{u},t)=x_i(t)+(1-u_i)\sum\nolimits_{j\in\mathcal L}A_{ij}x_j(t).
\end{equation}
In other words, if location $i$ is isolated, individuals in location $i$ are not allowed to leave their location, and so they can have interactions (and thus new contagions) only within their location. Since the control input can only reduce the infection force,  the result in 
Proposition~\ref{prop:invariance} remains valid, guaranteeing well-definedness of the controlled dynamics.

\begin{remark}In a more general setting, the control action in \eqref{eq:control} can be time-varying, i.e., having $\vect u(t)$. However, as we shall see in the following, the time-varying nature of interventions can be incorporated into the framework by defining an optimal control problem over a fixed time-window, within a rolling horizon scheme, ultimately inducing a nonlinear model-predictive control scheme~\cite{Mayne2014}.
\end{remark}

\section{Problem Statement}\label{sec:problem}

Designing optimal control policies using the interventions described in \eqref{eq:control} is nontrivial due to the inherent nonlinearity of the dynamical system and the discrete nature of the control actions, which naturally induce a binary optimization problem, which is often an instance of an NP-hard problem~\cite{Kochenberger2014}. Indeed, if one is interested only in the long-term behavior, it is known that the local stability of the disease-free equilibrium of the system is determined by the largest eigenvalue (in modulus) of the matrix $\mat I+\diag(1-\vect{u})\mat{A}$. Hence, optimizing the long-term behavior of the system can be mapped into the graph-theoretic problem of minimizing the spectral radius of a (weighted) adjacency matrix by edge removal, which is NP-hard~\cite{vanmieghem2011,Preciado2009}.



However, even the heuristics obtained with these graph-theoretic methods might be sub-optimal in the transient phase, and thus can be ineffective in curbing the epidemic curve during an outbreak. Moreover, it is important to keep in mind that enacting intervention policies that restrict mobility might be optimal from an healthcare viewpoint, but could potentially lead to tragic economic and social consequences, as it was observed during the COVID-19 pandemic. For this reason, in this paper, we investigate the problem of controlling the transient of an epidemic process, with a trade-off between the beneficial impact to the healthcare system and the economic and social losses.

We consider a time-window $T\in\mathbb Z_{>0}$ and we define the following cost function:
\begin{equation}\label{eq:cost}
    f(\vect{x},\vect{u})=\sum\nolimits_{i\in\mathcal L}\sum\nolimits_{t=1}^T x_i(t)+\gamma\sum\nolimits_{i\in\mathcal L} n_iu_i,
\end{equation}
which accounts for the healthcare cost of an epidemic outbreak (proportional to the total number of infections) and the cost associated with implementing the interventions (proportional to the total population that is impacted by the control action). The parameter $\gamma\geq 0$ can be used to weight the importance of the second contribution: larger values of $\gamma$ are associated with giving more importance to the social and economic aspects in the trade-off with healthcare consequences of the outbreak. 
At this point, we formalize the optimal control problem, as follows.

\begin{problem}\label{pr:problem}
    The optimal control policy is the solution
    \begin{equation}
    \begin{array}{rl}
       \vect{u^*}= &\arg\min_{\vect{x},\vect{u}} f(\vect{x},\vect{u})  \\
      \text{s.t.}   &\eqref{eq:sis}\text{ or }\eqref{eq:sir}\text{, and }\eqref{eq:alpha_control},
    \end{array}
\end{equation} 
where the first constraint is \eqref{eq:sis} or \eqref{eq:sir}, depending on the characteristics of the disease under investigation. 
\end{problem}

\section{QUBO Formulation}\label{sec:qubo}

We study Problem~\ref{pr:problem} and, for $T=2$, we derive a QUBO formulation, which is key for solving it in an efficient way using different techniques, including quantum computing~\cite{Kochenberger2014,Volpe2025}. We start observing that, given the initial conditions $\vect{\bar x}=[x_1(0),\dots,x_M(0)]$ (and also $\vect{\bar y}$ for the SIR model), the evolution of  $\vect{x}$ is fully determined by the two constraints in 
 \eqref{eq:sis} and \eqref{eq:alpha_control} and thus, ultimately, by  the control $\vect{u}$. Hence,  \eqref{eq:cost} can be written  in terms of $\vect{u}$, which is the only decision variable of the optimization problem. 

In general, being \eqref{eq:sis} and  \eqref{eq:sir} affine in the control input $\vect u$, its recursive use to express $\vect x(t)$ would yield a polynomial of order $t$. Hence, the intuition would suggest that the cost function in \eqref{eq:cost} can be written as a $T$th order polynomial. Here, we consider the scenario of $T=2$, for which we can express Problem~\ref{pr:problem} as a QUBO problem.

\begin{theorem}\label{theo}
The solution $\vect{u^*}$ of Problem~\ref{pr:problem} for the SIS model in \eqref{eq:sis} with $T=2$ is equal to $\vect{u^*}={\bf 1}-\vect{z^*}$, where $\vect z^*$ is the solution of the QUBO problem:
\begin{equation}\label{eq:cubo}
   \vect z^*=\arg\hspace{-.3cm}\min_{\vect u\in\{0,1\}^M}\sum\nolimits_{i\in\mathcal L}P_iz_i+\sum\nolimits_{i\in\mathcal L, j\in\mathcal L\setminus\{i\}}\hspace{-.2cm}Q_{ij}z_iz_j
\end{equation}
with
\begin{subequations}\begin{equation}\label{eq:cost_sis}
\begin{array}{l}
P_i=\displaystyle\lambda \Big(1\hspace{-.04cm}-\hspace{-.04cm}\frac{\bar x_i}{n_i}\Big)\sum_{j\in\mathcal L}A_{ij}\bar x_j+\lambda^2(1\hspace{-.04cm}-\hspace{-.04cm}\frac{\bar x_i}{n_i}\Big)^2\Big(\sum_{j\in\mathcal L}A_{ij}\bar x_j\Big)^2\\
\,\,\,\displaystyle+\lambda^2\Big[1-\frac{\bar x_i}{n_i}\Big(1-\mu+\lambda\Big(1-\frac{\bar x_i}{n_i}\Big)\Big]\Big(1-\frac{\bar x_i}{n_i}\Big)\sum_{j\in\mathcal L}A_{ij}\bar x_j\\
\,\,\,\displaystyle+\lambda\Big(1-\frac{\bar x_i}{n_i}\Big)\Big[1-\mu+\lambda\Big(1-\frac{\bar x_i}{n_i}\Big)\Big]\bar x_i\sum_{j\in\mathcal L}A_{ij}\bar x_j- \gamma n_i\\
\,\,\,\displaystyle+\lambda\Big[1-\frac{\bar x_i}{n_i}\Big(1-\mu+\lambda\Big(1-\frac{\bar x_i}{n_i}\Big)\Big)-\frac{\lambda}{n_i}\Big(1-\frac{\bar x_i}{n_i}\Big)\\
\qquad\displaystyle\cdot\sum\nolimits_{j\in\mathcal L}A_{ij}\bar x_j\Big]\sum\nolimits_{j\in\mathcal L}A_{ij}\Big[1-\mu+\lambda\Big(1-\frac{\bar x_j}{n_j}\Big)\Big]\bar x_j\end{array}\end{equation}
\begin{equation}\begin{array}{l}
    Q_{ij}=\displaystyle\lambda^2A_{ij}\Big[1-\frac{\bar x_i}{n_i}\Big(1-\mu+\lambda\Big(1-\frac{\bar x_i}{n_i}\Big)\Big)\\
   \displaystyle-\frac{\lambda}{n_i}\Big(1-\frac{\bar x_i}{n_i}\Big)\hspace{-.1cm}\sum\nolimits_{k\in\mathcal L} \hspace{-.3cm}A_{ik}\bar x_k\Big] \hspace{-.1cm} \Big(1-\frac{\bar x_j}{n_j}\Big)\sum\nolimits_{k\in\mathcal L} \hspace{-.3cm}A_{jk}\bar x_k.
\end{array}
\end{equation}
\end{subequations}
\end{theorem}
\begin{proof}
The result is obtained by using \eqref{eq:sis} recursively. In fact, for $t=1$, \eqref{eq:sis} becomes
\begin{equation}\label{eq:s1}
      x_i(1)=(1-\mu)\bar x_i+\lambda\Big(1-\frac{\bar x_i}{n_i}\Big)\Big(\bar x_i+(1-u_i)\sum\nolimits_{j\in\mathcal L}A_{ij}\bar x_j\Big).
\end{equation}
Then, for $t=2$, \eqref{eq:sis} becomes
\begin{equation}\label{eq:s2}
\begin{array}{l}
      x_i(2)=\displaystyle(1-\mu) x_i(1)+\lambda\Big(1-\frac{ x_i(1)}{n_i}\Big)\\
      \qquad\qquad\cdot\displaystyle\Big[ x_i(1)+(1-u_i)\sum\nolimits_{j\in\mathcal L}A_{ij} x_j(1)\Big].
      \end{array}
\end{equation}
The expression of $x_i(1)$ in \eqref{eq:s1} is a first-degree polynomial in $\vect u$. Inserting it into \eqref{eq:s2}, we obtain 
\begin{equation}\label{eq:s22}
\begin{array}{l}
      x_i(2)=\displaystyle(1-\mu)^2 \bar x_i+(1-\mu)\lambda\Big(1-\frac{ x_i}{n_i}\Big)\bar x_i\\
      +\displaystyle (1-\mu)\lambda(1-u_i)\Big(1-\frac{ x_i}{n_i}\Big)\sum\nolimits_{j\in\mathcal L}A_{ij}\bar x_j\\
      +\displaystyle\lambda \Big[1\hspace{-.04cm}-\hspace{-.04cm}\frac{\bar x_i}{n_i}\Big(1\hspace{-.04cm}-\hspace{-.04cm}\mu\hspace{-.04cm}+\hspace{-.04cm}\lambda\Big(1\hspace{-.04cm}-\hspace{-.04cm}\frac{ x_i}{n_i}\Big) \Big)\Big]\Big(1\hspace{-.04cm}-\hspace{-.04cm}\mu\hspace{-.04cm}+\hspace{-.04cm}\lambda\Big(1\hspace{-.04cm}-\hspace{-.04cm}\frac{ x_i}{n_i}\Big) \Big)\bar x_i\\
      +\displaystyle\lambda^2 \Big[1-\frac{\bar x_i}{n_i}\Big(1-\mu+\lambda\Big(1-\frac{ x_i}{n_i}\Big) \Big)\Big](1-u_i)\Big(1-\frac{ x_i}{n_i}\Big)\\
       \displaystyle\cdot\sum\nolimits_{j\in\mathcal L}A_{ij}\bar x_j+(1-u_i)\lambda^2\Big(1-\frac{ x_i}{n_i}\Big)^2\Big(\sum\nolimits_{j\in\mathcal L}A_{ij}\bar x_j\Big)^2\\
      +\displaystyle(1-u_i)\lambda\Big(1-\frac{\bar x_i}{n_i}\Big)\Big(1-\mu+\lambda\Big(1-\frac{ x_i}{n_i}\Big) \Big)\bar x_i\sum_{j\in\mathcal L}A_{ij}\bar x_j\\
      +\displaystyle	(1-u_i)\lambda\Big[1-\frac{\bar x_i}{n_i}\Big(1-\mu+\lambda\Big(1\hspace{-.04cm}-\hspace{-.04cm}\frac{\bar x_i}{n_i}\Big)\Big)+\lambda\Big(1\hspace{-.04cm}-\hspace{-.04cm}\frac{\bar x_i}{n_i}\Big)\\
      \displaystyle\cdot\sum\nolimits_{j\in\mathcal L}A_{ij}\bar x_j\Big]\Big[\sum\nolimits_{j\in\mathcal L}A_{ij}\Big(1-\mu+\lambda\Big(1-\frac{\bar x_j}{n_j}\Big)\Big)\bar x_j\Big]+\\
\displaystyle\lambda^2\hspace{-.04cm}\Big[1\hspace{-.04cm}-\hspace{-.04cm}\frac{\bar x_i}{n_i}\Big(1\hspace{-.04cm}-\hspace{-.04cm}\mu\hspace{-.04cm}+\hspace{-.04cm}\lambda\Big(1\hspace{-.04cm}-\hspace{-.04cm}\frac{\bar x_i}{n_i}\Big)\hspace{-.04cm}\Big)\hspace{-.04cm}+\hspace{-.04cm}\lambda\Big(1\hspace{-.04cm}-\hspace{-.04cm}\frac{\bar x_i}{n_i}\Big)\sum\nolimits_{j\in\mathcal L}\hspace{-.04cm}A_{ij}\bar x_j\hspace{-.04cm}\Big]\\
      \displaystyle\cdot\sum\nolimits_{j\in\mathcal L}A_{ij}(1-u_i)(1-u_j)\Big(1-\frac{\bar x_j}{n_j}\Big)\sum\nolimits_{k\in\mathcal L} A_{jk}\bar x_k,
      \end{array}
      \end{equation}
which is a second-degree polynomial in the control variable vector $\vect u$. Note that in the computations we use the fact that, being $u_i\in\{0,1\}$, then
$u_i^2=u_i$. Then, inserting \eqref{eq:s1} and \eqref{eq:s22} into  \eqref{eq:cost}, we obtain a second-degree polynomial expression in $\vect u$ for the cost function. Finally, by performing the change of variables $z_i=1-u_i$, we obtain the expression in \eqref{eq:cost_sis}, modulo some constant terms that have no impact on the result of the $\arg\min$ problem. Hence, the solution $\vect z^*$ of \eqref{eq:cubo} coincides with the solution of Problem~\ref{pr:problem}, with the change of variables $\vect u={\bf 1}-\vect z$. 
\end{proof}

The same approach can be used for the SIR model in \eqref{eq:sir}, for which the equivalent QUBO formulation in \eqref{eq:cubo} can be used to solve Problem~\ref{pr:problem}, with slightly different expressions for the coefficients $P_i$ and $Q_{ij}$ due to the change in the dynamics, which are reported in the following, with the proof omitted due to space constraints.

\begin{theorem}\label{theo2}
The solution $\vect{u^*}$ of Problem~\ref{pr:problem} for the SIR model in \eqref{eq:sir} with $T=2$ is equal to $\vect{u^*}={\bf 1}-\vect{z^*}$, where $\vect z^*$ is the solution of the QUBO problem in \eqref{eq:cubo} with
\begin{subequations}
\begin{equation}\label{eq:cost_sir}
\begin{array}{l}
 P_i=\displaystyle\lambda \Big(1-\frac{\bar x_i+\bar y_i}{n_i}\Big)\sum\nolimits_{j\in\mathcal L}A_{ij}\bar x_j- \gamma n_i\\
\,\,\,\displaystyle+\lambda^2(1-\frac{\bar x_i+\bar y_i}{n_i}\Big)^2\Big(\sum\nolimits_{j\in\mathcal L} A_{ij}\bar x_j\Big)^2\\
\,\,\,\displaystyle+\lambda^2\Big[1-\frac{\bar y_i}{n_i}-\frac{\bar x_i}{n_i}\Big(1-2\mu+\lambda\Big(1-\frac{\bar x_i+\bar y_i}{n_i}\Big)\Big]\\\displaystyle
\quad\cdot\Big(1-\frac{\bar x_i+\bar y_i}{n_i}\Big)\hspace{-.1cm}\sum\nolimits_{j\in\mathcal L} \hspace{-.2cm}A_{ij}\bar x_j+\lambda\Big(1-\frac{\bar x_i+\bar y_i}{n_i}\Big)\\\displaystyle
\quad \cdot\Big[1-2\mu+\lambda\Big(1-\frac{\bar x_i+\bar y_i}{n_i}\Big)\Big]\bar x_i\sum\nolimits_{j\in\mathcal L} A_{ij}\bar x_j\\
\,\,\,\displaystyle+\lambda\Big[1-\frac{\bar y_i}{n_i}-\frac{\bar x_i}{n_i}\Big(1-2\mu+\lambda\Big(1-\frac{\bar x_i+\bar y_i}{n_i}\Big)\Big)\\
\,\,\,\displaystyle-\frac{\lambda}{n_i}\Big(1-\frac{\bar x_i+\bar y_i}{n_i}\Big)A_{ij}\bar x_j\Big]\\\displaystyle
\quad\cdot\sum\nolimits_{j\in\mathcal L} A_{ij}\Big[1-\mu+\lambda\Big(1-\frac{\bar x_j+\bar y_j}{n_j}\Big)\Big]\bar x_j\end{array}\end{equation}
\begin{equation}\begin{array}{l}
Q_{ij}=\displaystyle\lambda^2A_{ij}\Big[1-\frac{\bar y_i}{n_i}-\frac{\bar x_i}{n_i}\Big(1-2\mu+\lambda\Big(1-\frac{\bar x_i}{n_i}-\frac{\bar y_i}{n_i}\Big)\Big)\\
    \,\,\,\displaystyle-\frac{\lambda}{n_i}\Big(1-\frac{\bar x_i+\bar y_i}{n_i}\Big)\sum_{k\in\mathcal L}A_{ik}\bar x_k\Big] \Big(1-\frac{\bar x_j+\bar y_i}{n_j}\Big)\sum_{k\in\mathcal L}A_{jk}\bar x_k.
\end{array}
\end{equation}
\end{subequations}
\end{theorem}

\begin{remark}\label{rem:T}
    The technique used in Theorem~\ref{theo} to cast Problem~\ref{pr:problem} as a QUBO problem can be directly used only for $T=2$. For $T>2$, the same approach would lead to a higher-degree polynomial. Different techniques may be used in these scenarios to deal with the problem, including exact methods through the addition of auxiliary variables to replace the higher-order terms, and approximated methods based on the recursive use of Taylor expansions.
\end{remark}


\begin{figure*}
    \centering
\subfloat[Regions, 2020/03/08]{\includegraphics[width=.43\columnwidth]{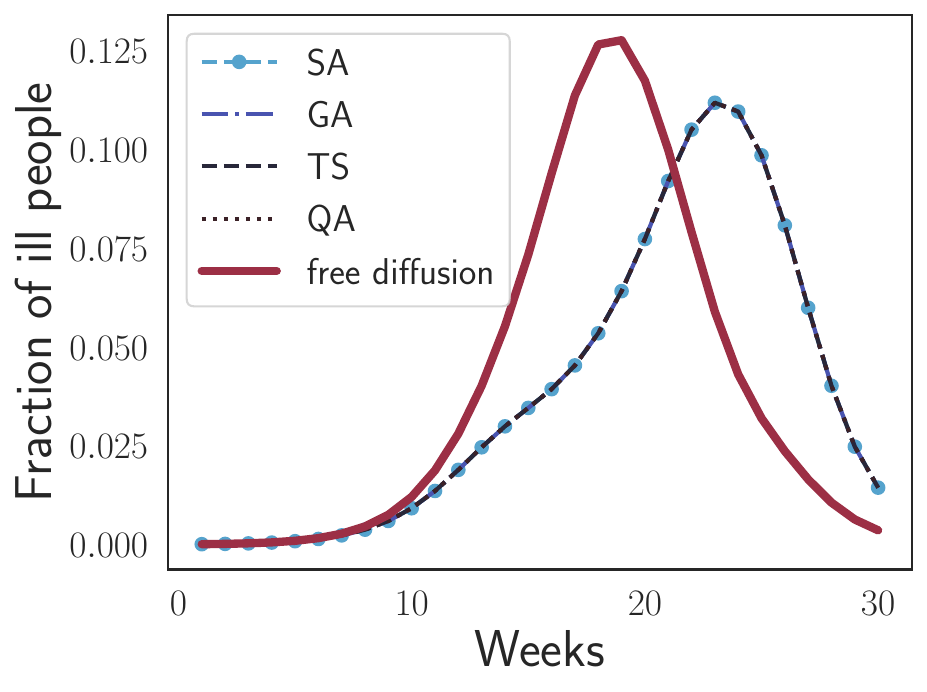}}\qquad
\subfloat[Regions, 2020/03/10]{\includegraphics[width=.43\columnwidth]{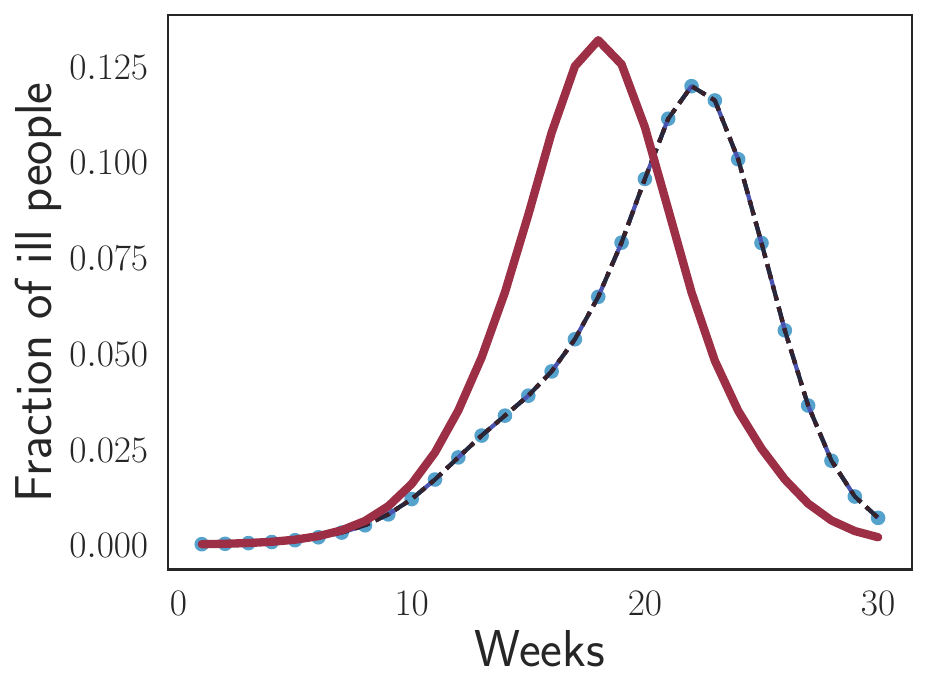}}\qquad\subfloat[Regions, 2020/01/29]{\includegraphics[width=.43\columnwidth]{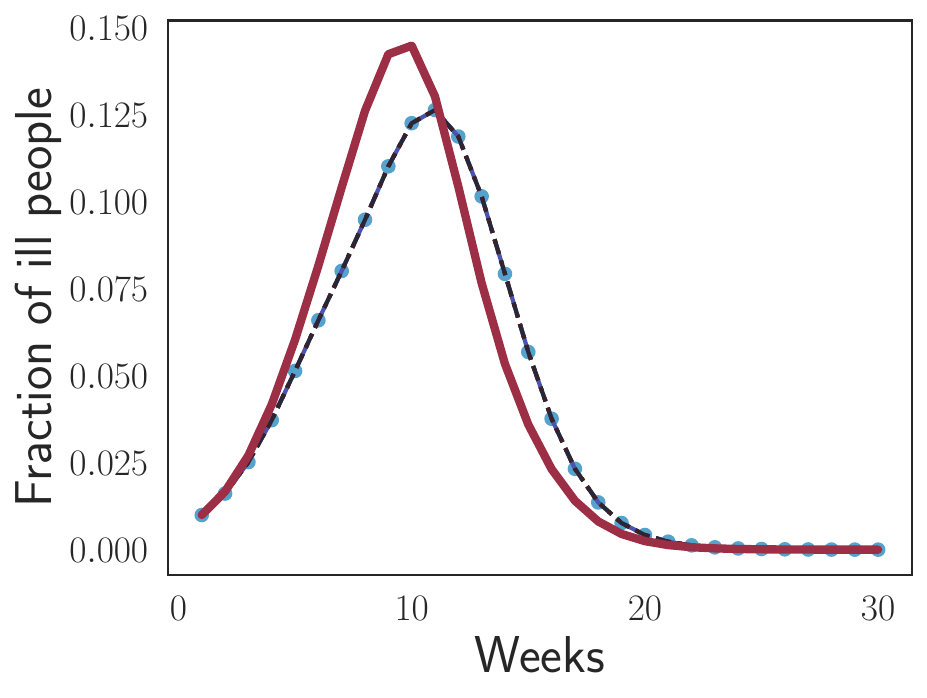}}\qquad\subfloat[Regions, 2021/12/20]{\includegraphics[width=.43\columnwidth]{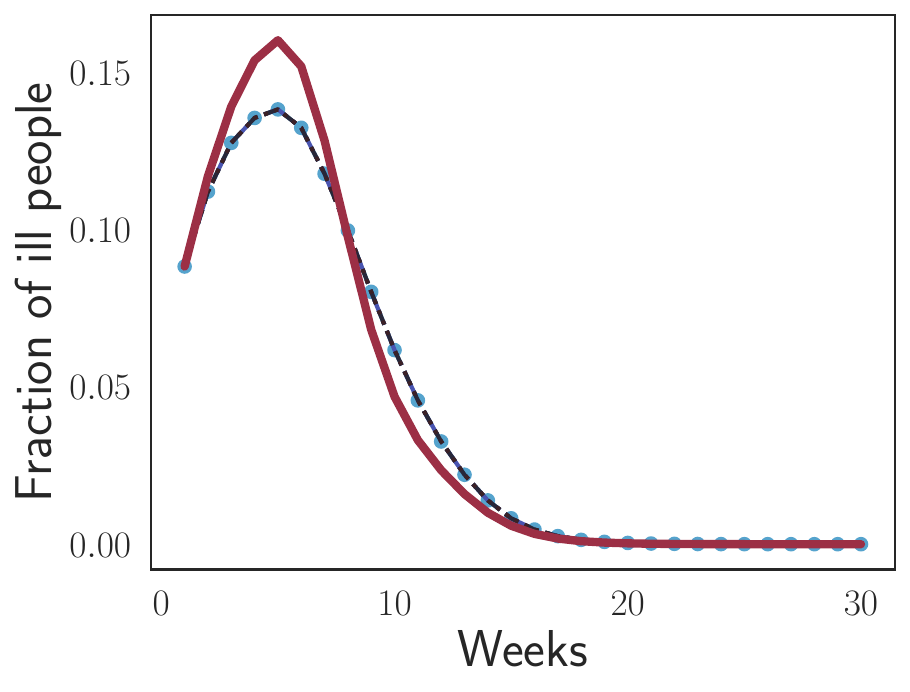}}\\[1ex]\subfloat[Provinces, 2020/03/08]{\includegraphics[width=.43\columnwidth]{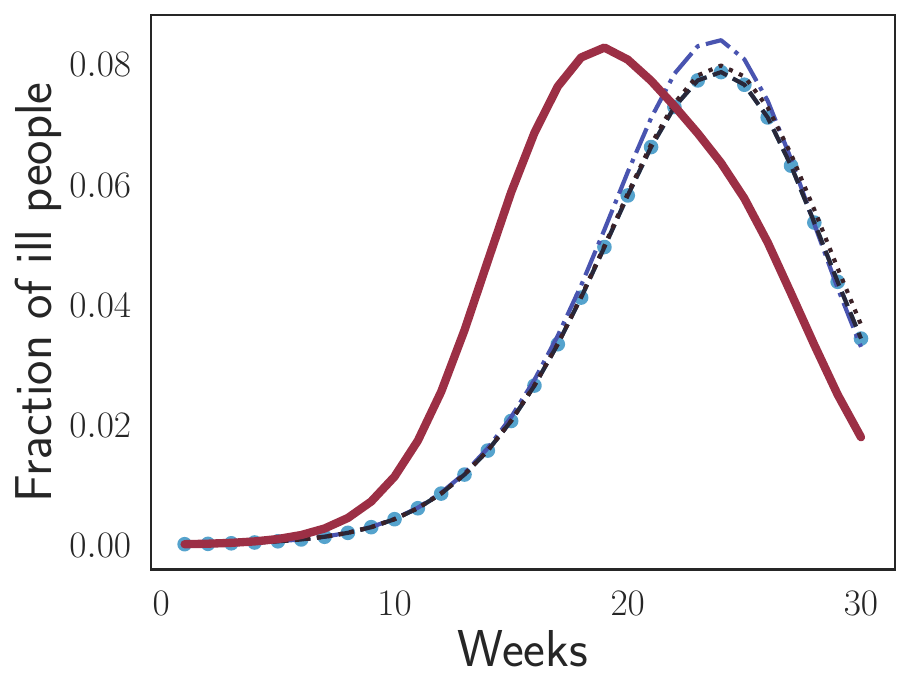}}
\qquad\subfloat[Provinces, 2020/03/10]{\includegraphics[width=.43\columnwidth]{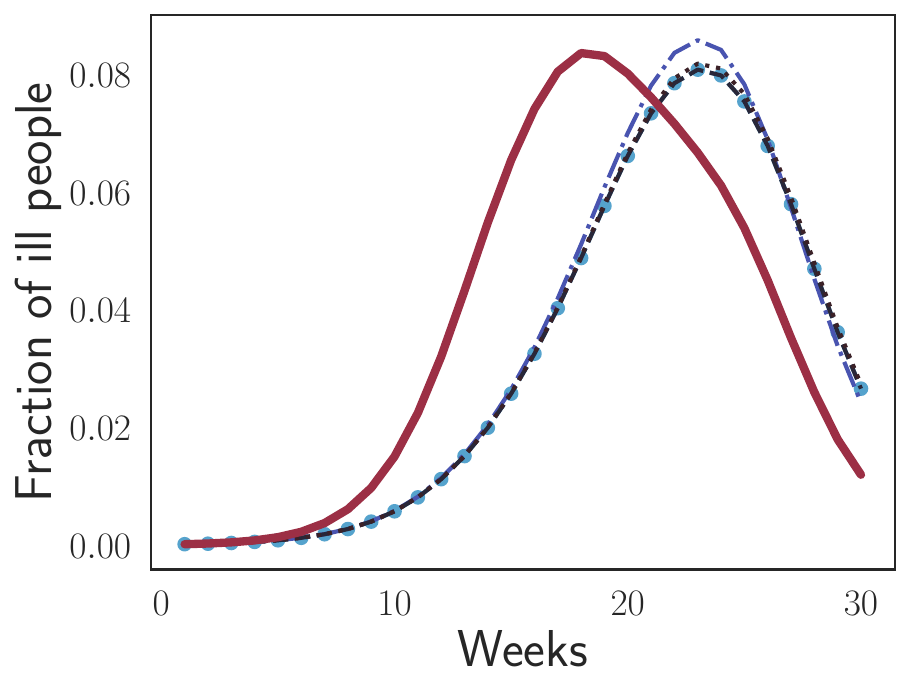}}\qquad\subfloat[Provinces, 2020/01/29]{\includegraphics[width=.43\columnwidth]{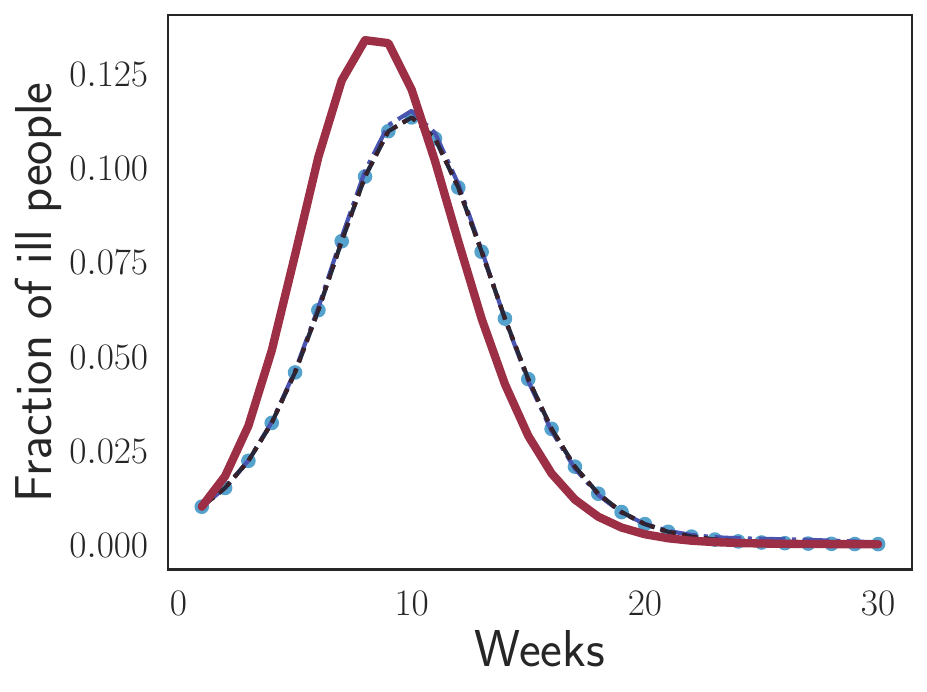}}\qquad\subfloat[Provinces, 2021/12/20]{\includegraphics[width=.43\columnwidth]{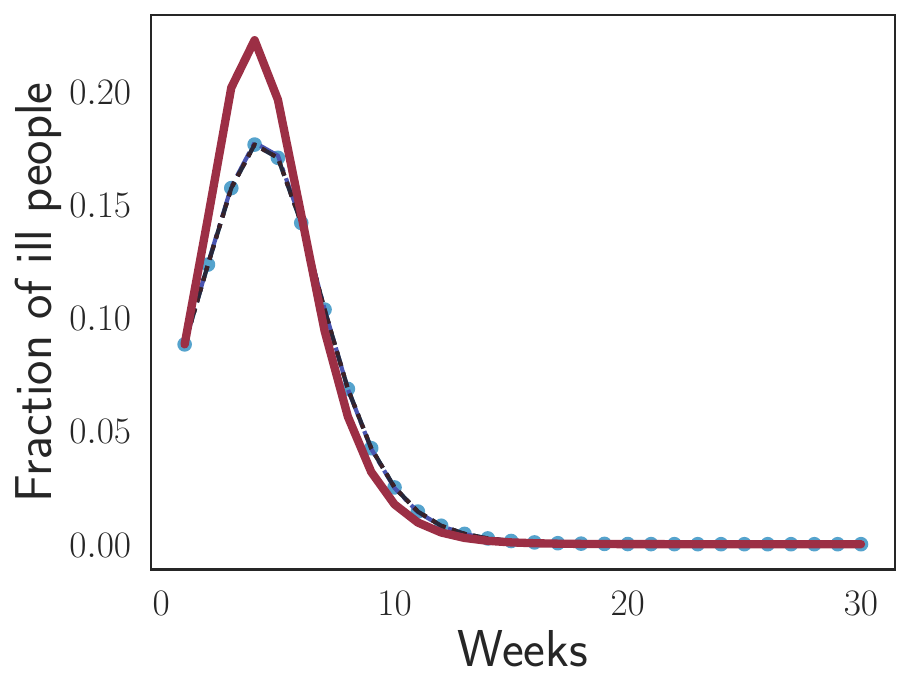}}
    \caption{In (a--d) and (e--h), we report the plots obtained at the granularity of regions and provinces, respectively, with different starting dates. }
    \label{fig:sim}
\end{figure*}

\begin{table}
    \centering
    \setlength{\tabcolsep}{1.5pt}
    \resizebox{0.4\textwidth}{!}{
        \begin{tabular}{c|c|cc|cc|cc|cc}
            \multirow{2}{*}{\textbf{Data}} & \multirow{2}{*}{\textbf{n}} & \multicolumn{2}{c|}{\textbf{SA}} & \multicolumn{2}{c|}{\textbf{TS}} & \multicolumn{2}{c|}{\textbf{GA}} &  \multicolumn{2}{c}{\textbf{QA}} \\
            & &  \textbf{p [\%]} & \textbf{a [\%]} &\textbf{p [\%]} & \textbf{a [\%]} & \textbf{p [\%]} & \textbf{a [\%]} & \textbf{p [\%]} & \textbf{a [\%]}
            \\
            \hline
            2020/03/08 & 21 & \bf 12.40 & \bf 3.48 & \bf 12.40 & \bf 3.48 & \bf  12.40 & \bf 3.48 & \bf 12.40 & 2.56 
            \\
            2020/03/10 & 21 & \bf 9.10 &\bf 2.63 &\bf 9.10 & \bf 2.63 &\bf 9.10 & \bf 2.63 & \bf 9.10 & 2.20 
            \\
            2020/10/29 & 21 & \bf 12.77 & 1.93 &\bf  12.77 & 1.93 &\bf   12.77 & 1.93&\bf   12.77 &\bf  3.00 
            \\
            2021/12/20 & 21 & \bf 13.73 & 1.72 & \bf  13.73 & 1.72 & \bf  13.73 & 1.72 & \bf  13.73 &\bf  10.05 
            \\
\hline
            2020/03/08 & 107 &\bf  4.93 &\bf 17.96 & \bf   4.93 &\bf 
 17.96  &  -1.48 & 14.31 &  3.61 & 16.66 
            \\
            2020/03/10 & 107 &\bf  3.37 &\bf  14.79 &\bf   3.37 &\bf  14.79  & -2.61 & 11.99 & 2.15 & 13.82 
            \\
            2020/10/29 & 107 &\bf  15.31 &\bf  9.02 &\bf   15.31 &\bf  9.02 &  14.02 & 7.74 &  15.29 &\bf  9.02 
            \\
            2021/12/20 & 107 &\bf  20.70 &\bf  7.54 &\bf  20.70 &\bf  7.54 & 20.31 & 7.34 &\bf  20.70 & \bf  7.54 
            \\
        \end{tabular}
    }
    \caption{Performance obtained with different methods in terms of reduction of the peak (p) and average infected/day (a) with respect to the uncontrolled dynamics. For each scenario, the best performance is highlighted using bold font.}
    \label{tab:SIRPicco}
\end{table}

\section{Numerical Results}\label{sec:results}

We demonstrate the effectiveness of our proposed approach on a real-world  case study, based on an epidemic outbreak spreading in Italy. Specifically, we consider a network where each node represents an Italian administrative division, and links (with their weights) represent individual mobility and travel between these divisions. Using publicly available data from the Italian National Institute of Statistics~\cite{Istat}, we are able to build networks with different granularity considering $n=21$ regions or $n=107$ provinces (second-level administrative divisions). The weight associated with links is set to be proportional to the fraction of population that commutes between two administrative divisions from~\cite{Istat}.

We consider the SIR epidemic model  in~\eqref{eq:sir}, calibrated on the  COVID-19 pandemic. Consistently, we set the parameters  following~\cite{parino2021rs} and using  Remark~\ref{rem:para},  making sure that the conditions in Proposition~\ref{prop:invariance} are satisfied. We initialize the simulations using the officially reported cases in each administrative division in the desired date, which varies across the simulations to explore the impact of the proposed control techniques in different  phases of the  outbreak~\cite{protezioneCivile}. 

In particular, for each one of the two levels of granularity (21 regions and 107 provinces), we selected four different initial dates, corresponding to four different phases of the epidemic outbreak: i) March 8, 2020; ii) March 10, 2020; iii) October 29, 2020; and iv) December 20, 2021. The first two dates coincide with early moments of the epidemic outbreak, which start from a moderately low number of cases. The third date captures a scenario of increasing cases due to a novel variant. The fourth date represents a later stage in which COVID-19 has become endemic, but a new wave was approaching due to higher levels of mobility in correspondence to Winter holidays.

\begin{table}[h]
    \centering
    \resizebox{0.45\textwidth}{!}{   \begin{tabular}{c|c|c|c|c|c}
           {\textbf{Data}} & {\textbf{n}} & {\textbf{SA}\,[\si{\milli\second}]} & {\textbf{TS}\,[\si{\milli\second}]} &{\textbf{GA}\,[\si{\milli\second}]} & {\textbf{QA}\,[\si{\milli\second}]} \\
            \hline
            2020/03/08 & 21 & 2.07 &  2112.88 & 55.92 & \bf 0.04  \\
            2020/03/10 & 21 & 1.94 &2103.59 & 49.50 &  \bf 0.04  \\
            2020/10/29 & 21 & 1.93  & 2110.27  & 51.32  & \bf 0.04  \\
            2021/12/20 & 21 & 1.89& 2119.29  & 68.24  & \bf 0.04  \\   \hline
            2020/03/08 & 107 & 784.40 &2124.64 & 508.54 &  \bf 0.04 \\
            2020/03/10 & 107 & 1072.64 &  2126.41 & 440.78 &   \bf 0.04  \\
            2020/10/29 & 107 & 824.93  & 2124.39 & 495.81  & \bf 0.04 \\
            2021/12/20 & 107 & 1306.50 & 2116.03  & 371.79 &  \bf 0.04 \\     
        \end{tabular}

    }
    \caption{Average solving time for the different methods. For each scenario, the best performance is highlighted using bold font.}
    \label{tab:SIR}
\end{table}

In our simulations, we implement the proposed control technique in a rolling horizon fashion, for a total simulation time-window of 30 weeks. Specifically, starting from the initial date of the simulation setting (which is denoted as $t=0$), we solve the control problem at time $t$ using the QUBO formulation in Theorem~\ref{theo2} over a time-window $T=2$, and we implement the optimal control policy $\vect {u^*}$, solution of Problem~\ref{pr:problem} for a single time step. Then, at time $t+1$, the same procedure is iterated, until the end of the simulation time-window is reached. 



The QUBO problem is solved using four different methods, including simulated annealing (SA) and tabu search (TS), both executed using D-Wave Solver, a genetic algorithm (GA) implemented via the \texttt{pymoo} optimization library, and quantum annealing (QA). These methods were selected to provide a representative comparison between classical metaheuristics and emerging quantum technologies~\cite{Hussain2018}. The results of our simulations are reported in Fig.~\ref{fig:sim}, which show the evolution of the epidemic curve over time under each control policy, compared with the uncontrolled baseline. The plots highlight how all methods significantly reduce the number of infected individuals over time.

Table~\ref{tab:SIRPicco} quantifies these improvements by reporting, for each solver, the relative reduction in the infection peak (denoted by p) and average number of infected/day (denoted by a) with respect to the uncontrolled dynamics; QA achieves performance comparable with those of the state-of-the-art classical methods in most scenarios. Importantly, such results are obtained with a remarkably smaller computational effort, as can be observed from Table~\ref{tab:SIR}. More details on the computational efficiency of the proposed method and further numerical simulations can be found in~\cite{Volpe2025qw}.

\section{Conclusion}\label{sec:conclusion}

In this article, we addressed the problem of optimally controlling the spread of epidemic diseases on networks through mobility bans. Building on classical discrete-time SIS and SIR models, we formulated a control problem to trade-off healthcare impact and socioeconomic costs of a mobility ban policy. We proved that this problem can be cast as a Quadratic Unconstrained Binary Optimization (QUBO) problem, allowing its solution through quantum annealing. We implemented our controller in a rolling-horizon  scheme, and we benchmarked the performance of four solvers ---including a quantum annealer--- on realistic epidemic scenarios calibrated on COVID-19 spreading in Italy, proving the effectiveness and computational efficiency of our approach.

The preliminary results presented in this paper pave the way for several promising avenues of future research. First, our approach was developed under the assumption two-step optimization time-window. In Remark~\ref{rem:T}, we observed that the use of recursive Taylor expansions is a possible approach to overcome this limitation. Future research should explore this research line. Second, while this paper focuses on control actions in terms of banning mobility between nodes, other types of interventions can be enacted during an epidemic outbreak. Extending our mathematical framework to introduce different control actions and compare their effectiveness is a key direction for future research. Third, the effectiveness of quantum computing in solving this optimization problem suggests that similar framework can be developed also in different contexts of controlling complex network systems.

\end{document}